\documentclass{LMCS}
\usepackage{hyperref}
\theoremstyle{plain}

\newcommand{\m}[1]{{\mbox{\uppercase {\bf {#1}}}}}
\newcommand{\sm}[1]{{\mbox{\scriptsize {\uppercase {\bf {#1}}}}}}
\newcommand{\vr}[1]{{\uppercase {\mathcal {#1}}}}

\newcommand{\pol}[1]{{\rm Pol\:\m #1}}

\newcommand{\con}{{\rm Con\:}}
\newcommand{\cn}[1]{{\con\m {#1}}}

\newcommand{\HSP}{{\sf HSP}}
\newcommand{\HS}{{\sf HS}}
\newtheorem{lm}[thm]{Lemma}
\newtheorem{prp}[thm]{Proposition}

\newtheorem{df}[thm]{Definition}

\newcommand{\jon}{J\'onsson\ }
\newcommand{\proj}{{\rm proj}}

\newcommand{\inv}{{\rm Inv}}

\def\doi{3 (2:6) 2007}
\lmcsheading%
{\doi}
{1--20}
{}
{}
{Oct.~18, 2006}
{Jun.~\phantom{0}8, 2007}
{}   

\begin{document}

\title{On tractability and congruence distributivity\rsuper*}

\author[E.~Kiss]{Emil Kiss\rsuper a}   %required
\address{{\lsuper a}Department of Algebra and Number Theory,
E{\"o}tv{\"o}s University, 1117 Budapest, P{\'a}zm{\'a}ny P{\'e}ter
s{\'e}t{\'a}ny 1/c,
Hungary} %required
\email{ewkiss@cs.elte.hu}  %optional
%\thanks{thanks 1, optional.}   %optional

\author[M.~Valeriote]{Matthew Valeriote\rsuper b} %optional
\address{{\lsuper b}Department of Mathematics and Statistics,
McMaster University, Hamilton, Ontario,
L8S 4K1, Canada}    %optional
\email{matt@math.mcmaster.ca}  %optional
%\thanks{thanks 2, optional.}    %optional

%% etc.

%% required for running head on odd and even pages, use suitable
%% abbreviations in case of long titles and many authors:

%% mandatory lists of keywords and classifications:
\keywords{constraint satisfaction problem, tractability,
universal algebra, congruence distributivity}
\subjclass{F.1.3, F.4.1}
\titlecomment{{\lsuper *}An extended abstract of this
paper has appeared in the Proceedings of the Twenty-First Annual
IEEE Symposium on Logic in Computer Science}
%%%%%%%%%%%%%%%%%%%%%%%%%%%%%%%%%%%%%%%%%%%%%%%%%%%%%%%%%%%%%%%%%%%%%%%%%%%

%% the abstract has to PRECEED the command \maketitle:
%% be sure not to issue the \maketitle command twice!

\begin{abstract}
\noindent Constraint languages that arise from finite algebras have recently
been an object of study, especially in connection with the Dichotomy
Conjecture of Feder and Vardi.  An important class of algebras are
those that generate congruence distributive varieties and included
among this class are lattices, and more generally, those algebras
that have near-unanimity term operations. An algebra will generate a
congruence distributive variety if and only if it has a sequence of
ternary term operations, called J\'onsson terms, that satisfy
certain equations.

We prove that constraint languages consisting of relations that are
invariant under a short sequence of J\'onsson terms are tractable by
showing that such languages have bounded relational width.
\end{abstract}

\maketitle

\section{Introduction}

The Constraint Satisfaction Problem (CSP) provides a framework for
expressing  a wide class of combinatorial problems.  Given an
instance of the CSP, the aim is to determine if there is a way to
assign values from a fixed domain to the variables of the instance
so that each of its constraints is satisfied.  While the entire
collection of CSPs forms an {\bf NP}-complete class of problems, a
number of subclasses have been shown to be tractable (i.e., to lie
in {\bf P}).  A major focus of research in this area is to determine
the subclasses of the CSP that are tractable.

One way to define a subclass of the CSP is to restrict the
constraint relations that occur in an instance to a given finite set
of relations over a fixed, finite domain, called a constraint
language. A central problem is to classify the constraint languages
that give rise to tractable subclasses of the CSP.  Currently, all
constraint languages that have been investigated have been shown to
give rise to a subclass of the CSP that is either {\bf NP}-complete
or in {\bf P}.  It is conjectured in \cite{feder-vardi-short} that
this dichotomy holds for all subclasses arising from finite
constraint languages.

In some special cases, the conjectured dichotomy has been verified.
For example, the work of Schaefer \cite{schaefer} and of Bulatov
\cite{bulatov-3-element} establish this over domains of sizes 2 and
3 respectively.  For constraint languages over larger domains a
number of significant results have been obtained \cite{bjk,
bulatov-conservative, dalmau-gmm}.

One method for establishing that the subclass of the CSP associated
with a finite constraint language is tractable is to establish a
type of local consistency property for the instances in the
subclass. In \cite{feder-vardi} Feder and Vardi  introduce a notion
of the width of a constraint language and show that  languages of
bounded width give rise to  tractable subclasses of the CSP.  There
is a natural connection between these subclasses of the CSP and
definability within Datalog.

In work by Jeavons and his co-authors an approach to classifying the
tractable constraint languages via algebraic methods has been
proposed and applied with great success \cite{bjk}.  In essence,
their work allows one to associate a finite algebraic structure to
each constraint language and then to analyze the complexity of the
corresponding subclass of the CSP in purely algebraic terms.

In this paper, we employ the algebraic approach to analyzing
constraint languages and with it are able to identify a new, general
class of tractable constraint languages. These languages arise from
finite algebras that generate congruence distributive varieties, or
equivalently, that have a sequence of special term operations,
called \jon terms, that satisfy certain equations. Theorem
\ref{main-result} establishes the tractability of these languages by
showing that they are of bounded width. Related to our result is the
theorem of Jeavons, Cohen, and Cooper in \cite{jeavons-cohen-cooper}
that establishes the tractability of constraint languages that arise
from another class of finite algebras that generate congruence
distributive varieties.  These algebras are equipped with a special
term operation called a near unanimity operation.  Dalmau
\cite{dalmau-gmm} provides an alternate proof of their result.

\section{Preliminaries}

In this  section we  introduce the necessary terminology and results
on the CSP and from universal algebra that will be needed to prove
the main result (Theorem \ref{main-result}) of this paper.

In the following discussion we will employ standard terminology and
notation when dealing with $n$-tuples and relations over sets. In
particular, if $\vec{a}$ is a tuple over the sequence of domains
$A_i$, $1 \le i \le n$, (i.e., is a member of $\prod_{1 \le i \le
n}A_i$) and $I$ is a subset of $\{1,2, \ldots, n\}$ then
$\proj_I(\vec{a})$ denotes the tuple $( a_i\,:\, i \in I) \in
\prod_{i \in I}A_i$ over the sequence of domains $(A_i\,:\, i \in
I)$ and is called the restriction (or the projection) of $\vec{a}$
to $I$. We extend this projection function to arbitrary relations
over the $A_i$. The $i$th element of the tuple $\vec{a}$ will be
denoted by $\vec{a}(i)$.

For $R$ and $S$ binary relations on a set $A$, we define the
relational product of $R$ and $S$, denoted $R\circ S$, to be the
binary relation consisting of all pairs $(a,b)$ for which there is
some $c$ with $(a,c) \in R$ and $(c,b) \in S$.

\subsection{The Constraint Satisfaction Problem}

\begin{df}\label{csp-def}  An instance of the constraint satisfaction problem is a
triple $P = (V, A, \vr c)$ with
\begin{itemize}
    \item $V$ a non-empty, finite set of variables,
    \item $A$ a non-empty, finite set (or domain),
    \item $\vr c$ a set of constraints $\{C_1, \ldots, C_q\}$
    where each $C_i$ is a pair $(\vec{s}_i, R_i)$ with
    \begin{itemize}
        \item $\vec{s}_i$ a tuple of variables of length $m_i$,
        called the scope of
        $C_i$, and
        \item $R_i$ an $m_i$-ary relation over $A$, called the constraint
        relation of $C_i$.
    \end{itemize}
\end{itemize}
\end{df}

Given an instance $P$ of the CSP we wish to answer the following
question:
\begin{quote} Is there a solution to $P$, i.e., is
there a function $f:V \rightarrow A$ such that for each $i \le q$,
the $m_i$-tuple $f(\vec{s}_i) \in R_i$?
\end{quote}

We say that two instances of the CSP having the same set of
variables and  the same domain are equivalent if they have the same
set of solutions.

 In general, the class of CSPs is {\bf NP}-complete
(see \cite{jeavons-cohen-cooper}), but by restricting the nature of
the constraint relations that are allowed to appear in an instance
of the CSP, it is possible to find natural subclasses of the CSP
that are tractable.

\begin{df}  Let $A$ be a domain and $\Gamma$
a set of finitary relations over $A$. CSP($\Gamma$) denotes the
collection of all instances of the CSP with domain $A$ and with
constraint relations coming from $\Gamma$.  $\Gamma$ is called the
constraint language of the class CSP($\Gamma$).
\end{df}

\begin{df}

Call a finite constraint language $\Gamma$  tractable if the class
of problems CSP($\Gamma$) is tractable (i.e., lies in {\bf P}).  If
$\Gamma$ is infinite and each finite subset $\Gamma'$ of $\Gamma$ is
 tractable then we say that $\Gamma$ is tractable.  If the entire
 class CSP($\Gamma$) is in {\bf P} then we say that $\Gamma$ is
 globally tractable.

$\Gamma$ is said to be {\bf NP}-complete if for some finite subset
$\Gamma'$ of $\Gamma$, the class of problems CSP($\Gamma'$) is {\bf
NP}-complete.
\end{df}

A key problem in this area is to classify the  (globally) tractable
constraint languages.  Note that in this paper we will assume that
${\bf P} \ne {\bf NP}$.  Feder and Vardi \cite{feder-vardi}
conjecture that every finite constraint language is either tractable
or is {\bf NP}-complete.

We will find it convenient to extend the above notions of instances
of the CSP and constraint languages to a multi-sorted setting. This
approach has been used on a number of occasions, in particular in
\cite{bulatov-2semi}.

\begin{df}\label{multi-sort}  A multi-sorted instance of the constraint satisfaction problem is a
pair $P = (\vr a, \vr c)$ where
\begin{itemize}
\item $\vr a = (A_1, A_2, \ldots, A_n)$ is a sequence of finite,
non-empty sets, called the domains of $P$, and
    \item $\vr c$ is a set of constraints $\{C_1, \ldots, C_q\}$
    where each $C_i$ is a pair $(S_i, R_i)$ with
    \begin{itemize}
        \item $S_i$ a non-empty subset of $\{1,2, \ldots, n\}$
        called the scope of
        $C_i$, and
        \item $R_i$ an $|S_i|$-ary relation over $(A_j\,:\, j \in S_i)$, called the constraint
        relation of $C_i$.
    \end{itemize}
\end{itemize}
\end{df}

In this case, a solution to $P$ is an $n$-tuple $\vec{a}$ over the
sequence $(A_i\,:\, 1\le i \le n)$ such that $\proj_{S_j}(\vec{a})
\in R_j$ for each $1 \le j \le q$.  Clearly, each standard instance
of the CSP can be expressed as an equivalent multi-sorted instance.
While the given definition of a multi-sorted instance of the CSP
does not allow for the repetition of variables within the scope of
any constraint, there is a natural extension of
Definition~\ref{csp-def} that allows this.  Note that there is a
very straightforward procedure to transform such an instance to an
equivalent one that conforms to Definition~\ref{multi-sort}.

\begin{df} A relation $R$ over the sets $A_i$, $1\le i \le n$, is
subdirect if for all $1 \le i \le n$, $\proj_{\{i\}}(R) = A_i$. We
call a multi-sorted instance $P$  of the CSP subdirect if each of
its constraint relations is.
\end{df}

In addition to the set of solutions of an instance of the CSP, one
can also consider partial solutions of the instance.

\begin{df}
For $P$ as in Definition \ref{multi-sort} and $I$ a subset of
$\{1,2,
\ldots, n\}$, the set of
 partial solutions of $P$ over $I$, denoted $P_I$, is the set of
solutions of the instance $P' = (\vr a', \vr c')$ where $\vr a' =
(A_i\,:\, i \in I)$ and $\vr c = \{C'_1, \ldots, C'_q\}$ with $C'_j
= (I \cap S_j, \proj_{(I\cap S_j)}(R_j))$ for $1 \le j \le q$.
\end{df}

Clearly if the set of partial solutions of an instance over some
subset of coordinates is empty then the instance has no solutions.

\begin{df} Let $\vr c$ be a finite set (or sequence) of finite, non-empty sets.  A
(multi-sorted) constraint language over $\vr c$ is a collection of
finitary relations over the sets in $\vr c$.  Given a multi-sorted
constraint language $\Gamma$ over $\vr c$, the class CSP($\Gamma$)
consists of all multi-sorted instances of the CSP whose domains come
from $\vr c$ and whose constraint relations come from $\Gamma$.
$\Gamma_{\vr c}$ denotes the set of all finitary relations over the
members of $\vr c$.
\end{df}

In a natural way, the notions of tractability and {\bf
NP}-completeness can be extended to multi-sorted constraint
languages.

\subsection{Algebras}

There are a number of standard sources for the basics of universal
algebra, for example \cite{Bu-Sa} and \cite{alvi}.  The books
\cite{Ho-Mck, clasen-valeriote} provide details on the more
specialized aspects of the subject that we will use in this paper.

\begin{df}  An algebra $\m a $ is a pair $(A, F)$ where $A$ is a
non-empty set and $F$ is a (possibly infinite) collection of
finitary operations on $A$.  The operations in $F$ are called the
basic operations of $\m a$.  A term operation of an algebra $\m a$
is a finitary operation on $A$ that can be obtained by repeated
compositions of the basic operations of $\m a$.
\end{df}

We assume some familiarity with the standard algebraic operations of
taking subalgebras, homomorphic images and cartesian products.  Note
that in order to sensibly take a homomorphic image of an algebra, or
the cartesian product of a set of algebras or to speak of terms and
equations of an algebra we need to have some indexing of the basic
operations of the algebras. Algebras that have the same indexing are
said to be similar (or of the same similarity type).

When necessary, we distinguish between an algebra and its underlying
set, or universe.  A subuniverse of an algebra $(A, F)$ is a subset
of $A$ that is invariant under $F$. Note that we allow empty
subuniverses but not algebras with empty universes.

\newcommand{\V}{\mbox{\sf V}}

\begin{df}  A variety of algebras is a collection of similar
algebras that is closed under the taking of cartesian products,
subalgebras and homomorphic images.  If $\vr k$ is a class of
similar algebras then $\V(\vr k)$ denotes the smallest variety that
contains $\vr k$.
\end{df}

\begin{thm}[Birkhoff]  A class $\vr v$ of similar algebras is a variety if
and only if $\vr v$ can be axiomatized by a set of equations.
\end{thm}

It turns out that for a class $\vr k$ of similar algebras, $\V(\vr
k) = \HSP(\vr k)$, i.e., the class of homomorphic images of
subalgebras of cartesian products of members of $\vr k$.

\begin{df}  Let $\m a$ be an algebra.
\begin{enumerate}
\item  An equivalence relation $\theta$ on $A$ is a congruence of
$\m a$ if it is invariant under the basic operations of $\m a$.
\item The congruence lattice of $\m a$, denoted $\con(\m a)$, is the
lattice of all congruences of $\m a$, ordered by inclusion.

\item $0_A$ denotes the congruence relation $\{(a,a): a \in A\}$
and $1_A$ denotes the congruence relation $\{(a,b): a, b \in A\}$,
the smallest and largest congruences of the algebra $\m a$,
respectively.

\item  An algebra $\m a$ is simple if $0_A$ and $1_A$ are its only
congruences.
\end{enumerate}
\end{df}

The congruence lattice of an algebra is a very useful invariant and
the types of congruence lattices that can appear in a variety govern
many properties of the algebras in the variety. One particularly
relevant and important property of congruence lattices is that  of
distributivity.

\begin{df} An algebra $\m a$ is said to be {congruence distributive}
if its congruence lattice satisfies the distributive law for
congruence meet and join. A class of algebras is congruence
distributive if all of its members are.
\end{df}

\begin{df} \label{JonEqu} For $k > 0$, we define $CD(k)$ to be the class
of all algebras $\m a$ that have a sequence of ternary term
operations $p_i(x,y,z)$, $0 \le i \le k$, that satisfies the
identities:
\begin{eqnarray*}
p_0(x,y,z) &=&x\\
p_k(x,y,z) &=&z\\
p_i(x,y,x) &= &x \mbox{ for all $i$}\\
p_i(x,x,y) &=& p_{i+1}(x,x,y) \mbox{ for all $i$ even}\\
p_i(x,y,y) &=& p_{i+1}(x,y,y) \mbox{ for all $i$ odd}
\end{eqnarray*}
\end{df}

A sequence of term operations of an algebra $\m a$ that satisfies
the above equations will be referred to as \jon terms of $\m a$. The
following celebrated theorem of \jon relates congruence
distributivity to the existence of \jon terms.

\begin{thm}[J\'onsson]  An algebra $\m a$ generates a congruence
distributive variety if and only if there is some $k > 0$ such that
$\m a$ is in $CD(k)$. In this case, all algebras in $\V(\m a)$ lie
in $CD(k)$.
\end{thm}

\begin{df}  For $k > 1$, define $\vr v_k$ to be the variety  of all algebras
that have as basic operations a sequence of $k+1$ ternary operations
$p_i(x,y,z)$, for $0 \le i \le k$, that satisfy the equations from
Definition~\ref{JonEqu}.
\end{df}

Note that an algebra is in $CD(1)$ if and only if it has size 1 and
is in $CD(2)$ if and only if it has a majority term operation (i.e.,
a term operation $m(x,y,z)$ that satisfies the equations $m(x,x,y) =
m(x,y,x) = m(y,x,x) = x$).

Some of the main results and conjectures dealing with the CSP can be
expressed in terms of Tame Congruence Theory, a deep theory of the
local structure of finite algebras developed by Hobby and McKenzie.
Details of this theory may be found in \cite{Ho-Mck} or
\cite{clasen-valeriote}. The connection between the CSP and Tame
Congruence Theory was made by Bulatov, Jeavons, and Krokhin
\cite{bjk} and we will touch on it in the next subsection. In this
paper we will only introduce some of the basic terminology of the
theory and will omit most details.

In Tame Congruence Theory, five local types of behaviour of  finite
algebras are identified and studied.  The five types are, in order:
\begin{quote}
\begin{enumerate}
\item the unary type,
\item the affine or vector-space type,
\item the 2 element Boolean type,
\item the 2 element lattice type,
\item the 2 element semi-lattice type.
\end{enumerate}
\end{quote}
We say that an algebra $\m a$ omits a particular type  if, locally,
the corresponding type of behaviour does not occur in $\m a$.  A
class of algebras $\vr c$ is said to omit a particular type  if all
finite members of $\vr c$ omit that type.

In \cite{Ho-Mck}, chapter 9,  characterizations of finite algebras
that generate varieties that omit the unary type or both the unary
and affine type are given.  The characterizations are similar to
that given by \jon of the congruence distributive varieties.  It
easily follows from the characterizations that if $\m a$ is a finite
algebra that generates a congruence distributive variety then the
variety omits both the unary and affine types.

To close this subsection we note a special property of the term
operations of the algebras in $\vr v_k$ for all $k > 1$.
\begin{df}
An $n$-ary operation $f(x_1, \ldots, x_n)$ on a set $A$ is
idempotent if for all $a \in A$, $f(a,a, \ldots, a) = a$ .  An
algebra is idempotent if all of its term operations are idempotent.
\end{df}

Note that idempotency is hereditary in the sense that if a function
is the composition of some idempotent operations then it too is
idempotent.  In another sense, if $\m a$ is idempotent then all
algebras in $\V(\m a)$ are idempotent, since this condition can be
described equationally.  Finally, note that \jon terms are
idempotent and so all algebras in $\vr v_k$ for $k > 1$ are
idempotent.

\subsection{Algebras and the CSP}

The natural duality between sets of relations (constraint languages)
over a set $A$ and sets of operations (algebras) on $A$ has been
studied by algebraists for some time.  Jeavons and his co-authors
\cite{jeavons} have shown how this link between constraint languages
and algebras can be used to transfer questions about tractability
into equivalent questions about algebras.  In this subsection we
present a concise overview of this connection.

\begin{df}  Let $A$ be a non-empty set.
\begin{enumerate}
\item
Let $R$ be an $n$-ary relation over  $A$ and $f(\bar x)$ an $m$-ary
function over $A$ for some $n$, $m \ge 0$. We say that $R$ is
invariant under $f$ and that $f$ is a polymorphism of $R$ if for all
$\vec{a}_i \in R$, for $1 \le i \le m$, the $n$-tuple $f(\vec{a}_1,
\ldots, \vec{a}_m)$,  whose $i$-th coordinate is equal to
$f(\vec{a}_1(i), \ldots, \vec{a}_m(i))$, belongs to $ R$.

\item  For $\Gamma$ a set of relations over $A$, $\pol(\Gamma)$
denotes the set of functions on $A$ that are polymorphisms of all
the relations in $\Gamma$.

\item For $F$ a set of finitary operations on $A$, $\inv(F)$ denotes
the set of all finitary relations on $A$ that are invariant under
all operations in $F$.

\item For $\Gamma$ a constraint language over $A$, $\langle
\Gamma \rangle$ denotes $\inv(\pol(\Gamma))$ and $\m a_\Gamma$
denotes the algebra $(A, \pol(\Gamma))$.

\item For $\m a = (A, F)$, an algebra over $A$, $\Gamma_{\sm a}$
denotes the constraint language $\inv(F)$.

\item We call a finite algebra $\m a$ tractable ({\bf NP}-complete) if the
constraint language $\Gamma_{\sm a}$ is.
\end{enumerate}
\end{df}

Note that if $\m a$ is an algebra, then $\inv(\m a)$ coincides with
the set of all subuniverses of finite cartesian powers of $\m a$.
Sets of relations of the form $\inv(\Gamma)$ for a set of relations
$\Gamma$ are known as relational clones. Equivalently, a set of
relations $\Lambda$ over a finite set $A$ is a relational clone if
and only if it is closed under definition by primitive positive
formulas (or conjunctive queries).

\begin{thm}(\cite{jeavons})\label{relclone}
Let $\Gamma$ be a constraint language on a finite set.  If $\Gamma$
is tractable then so is $\langle\Gamma \rangle$.  If $\langle\Gamma
\rangle$ is {\bf NP}-complete then so is $\Gamma$.
\end{thm}

In algebraic terms, Theorem \ref{relclone} states that a constraint
language $\Gamma$ is tractable (or {\bf NP}-complete) if and only if
the algebra $\m a_\Gamma$ is.  So, the problem of characterizing the
tractable constraint languages can be reduced to the problem of
characterizing the tractable finite algebras.  In a further step,
Bulatov, Jeavons and Krokhin \cite{bjk} provide a  reduction down to
idempotent algebras.  For this class of algebras, they propose the
following characterization of tractability.

\begin{conj} \label{tract-conj} Let $\m a$ be a finite idempotent algebra.  Then $\m
a$ is tractable if and only if the variety $\V(\m a)$ omits the
unary type.
\end{conj}

They show that when this condition fails, the algebra is {\bf
NP}-complete \cite{bjk}.  They also show that if $\m a$ is a finite,
idempotent algebra then $\V(\m a)$ omits the unary type if and only
if the class $\HS(\m a)$ does.  This conjecture has been verified
for a number of large classes of algebras.  For example, results of
Schaefer \cite{schaefer} and Bulatov \cite{bulatov-3-element}
provide a verification for algebras whose universes have size 2 and
3 respectively.

As noted in the introduction, one approach to proving the
tractability of a constraint language $\Gamma$ is to apply a  notion
of local consistency to the instances in CSP($\Gamma$) to determine
if the instances have  solutions.  We present a notion of width,
called  relational width, developed by Bulatov and Jeavons
\cite{bulatov-jeavons} that, for finite constraint languages, is
closely related to the notion of  width defined by Feder and Vardi
(see \cite{larose-width-notes, larose-zadori-width}).  In this paper
we will closely follow the presentation of relational width found in
\cite{bulatov-2semi}.

\begin{df} Let $\vr a = (A_1, \ldots, A_n)$ be a sequence of finite, non-empty sets, let
 $P= (\vr a, \vr c)$ be an instance of the CSP and let
$k > 0$.  We say that $P$ is $k$-minimal if:
\begin{enumerate}
\item  For each subset $I$ of $\{1, 2, \ldots, n\}$ of size at most $k$,
 there is some
constraint $(S, R)$ in $\vr c$ such that $I \subseteq S$, and
\item If $(S_1, R_1)$ and $(S_2, R_2)$ are constraints in $\vr c$
and $I \subseteq S_1 \cap S_2$ has size at most $k$ then
$\proj_I(R_1) = \proj_I(R_2)$.
\end{enumerate}
\end{df}

It is  not hard to show that the second condition of this definition
is equivalent to having  the set of partial solutions $P_I$ of $P$
equal to  $\proj_I(R_i)$ for all subsets $I$ of size at most $k$ and
all $i$ with $I \subseteq S_i$.

\begin{prp}\label{min-alg}
Let $\Gamma$ be a constraint language and $k > 0$.
There is a
polynomial time algorithm (the $k$-minimality algorithm)
 that converts a
given instance $P$ from CSP($\Gamma$)
into an equivalent $k$-minimal instance  $P'$ %(i.e., that has the same solution set)
from CSP($\langle\Gamma\rangle$).  In fact, if the arities of the
constraint relations of $P$ are bounded by an integer $m\ge k$ then
the arities of the constraint relations of $P'$ are also bounded by
$m$.
\end{prp}

\begin{proof}  See the discussion in Section 3.1 of
\cite{bulatov-2semi}.
\end{proof}

\begin{df}  Let $\Gamma$ be a constraint language and $k > 0$.  We
say that $\Gamma$ has relational width $k$ if for every instance $P$
from CSP($\Gamma$), $P$ has a solution if and only if the constraint
relations of  $P'$, the equivalent $k$-minimal instance produced by
the $k$-minimality algorithm, are all non-empty.
\end{df}

\begin{prp}\label{width-prop}
Let $\Gamma$ be a constraint language and $k > 0$.
\begin{enumerate}
\item If an instance $P$ of the CSP has a solution
then the constraint relations of all equivalent instances are
non-empty.

\item  If $\Gamma$ has relational width $k$ and $\Delta \subseteq
\Gamma$ then $\Delta$ also has relational width $k$.

\item If $\Gamma$ has relational width $k$ then every $k$-minimal
instance $P$ from CSP($\Gamma$) whose constraint relations are
non-empty has a solution.

\item If $\Gamma$  is of finite relational width then it is globally
tractable.

\item If every $k$-minimal instance from CSP($\langle\Gamma\rangle$)
whose constraint relations are non-empty has a solution then
$\Gamma$ has relational width $k$ and hence is globally tractable.

\item If $\Gamma$ is finite and $m \ge k$ is an upper bound on the
arities of the relations in $\Gamma$ then $\Gamma$ has relational
width $k$ if every $k$-minimal instance from
CSP($\langle\Gamma\rangle$) whose constraint relations are non-empty
and have arity $\le m$ has a solution.

\end{enumerate}

\end{prp}

\begin{proof}
Statement (4) follows from Proposition \ref{min-alg}, since if
$\Gamma$ has relational width $k$ and $P$ is an instance from
CSP($\Gamma$) then in order to determine if $P$ has a solution, it
suffices to test if $P'$, the equivalent $k$-minimal instance
produced by the $k$-minimality algorithm, has non-empty constraint
relations. Statements (5) and (6)  also follows from Proposition
\ref{min-alg} since the constraint relations of $P'$ belong
to $\langle \Gamma \rangle$ and their arities are no bigger than the
maximum of $k$ and the arities of the constraint relations of $P$.
\end{proof}

In the case where $\Gamma$  happens to be a relational clone (i.e.,
$\Gamma = \langle \Gamma \rangle$) it follows from statements (3)
and (5) of the previous proposition that $\Gamma$ has relational
width $k$ if and only if every $k$-minimal instance of CSP($\Gamma$)
whose constraint relations are all non-empty has a solution.  For
the most part, we are interested in this type of constraint language
in this paper.

We note that in \cite{larose-width-notes,larose-zadori-width} it is
shown that a finite constraint language has bounded relational width
if and only if it has bounded width in the sense of Feder-Vardi.
 The following conjecture is similar to Conjecture
\ref{tract-conj} and was proposed by Larose and Z\'adori
\cite{larose-zadori-width} for constraint languages of bounded
width.

\begin{conj}\label{width-conj}  Let $\m a$ be a finite idempotent
algebra.  Then $\m a$ is of bounded width if and only if $\V(\m a)$
omits the unary and affine types.
\end{conj}

In \cite{larose-zadori-width} Larose and Z\'adori verify one
direction of this conjecture, namely that if $\V(\m a)$ fails to
omit the unary or affine types then $\m a$ is not of bounded width.
Note that in \cite{bulatov}, Bulatov proposes a  conjecture that is
parallel to \ref{width-conj}.  Larose and the second author have
noted  that, as with the unary type, one need only check in $\HS(\m
a)$ to determine if $\V(\m a)$ omits the unary and affine types when
$\m a$ is finite and idempotent (see Corollary 3.2 of
\cite{valeriote-intprop} for a more general version of this).

The main result of this paper can be regarded as providing some
evidence in support of Conjecture \ref{width-conj}. Theorem
\ref{main-result} establishes that if $\m a$ is a finite member of
$CD(3)$ then any finite constraint language contained in
$\Gamma_{\sm a}$ is of bounded width and hence tractable.

\section{Algebras in $CD(3)$}\label{reductions}

Recall that the variety $\vr v_3$ consists of all algebras $\m a$
having four basic operations $p_i(x,y,z)$, $0 \le i \le 3$ that
satisfy the equations of Definition~\ref{JonEqu}.  Since the
equations dictate that $p_0$ and $p_3$ are projections onto $x$ and
$z$ respectively, they will play no role in the analysis of algebras
in $CD(3)$.

\subsection{J\'onsson ideals}

For $\m a$ an algebra in $\vr v_3$, define $x\cdot y$ to be the
binary term operation $p_1(x,y,y)$ of $\m a$.  Note that the \jon
equations imply that $x \cdot y = p_2(x,y,y)$ as well. This
``multiplication'' will play a crucial role in the proof of the main
theorem of this paper.

\begin{df}
For $X$ a subset of an algebra $\m b \in \vr v_3$  let $J(X)$ be the
smallest subuniverse $Y$ of $\m b$ containing $X$ and satisfying the
following closure property: if $x$ is in $Y$ and $u \in B$ then
$u\cdot x$ is also in $Y$.
\end{df}

We will call $J(X)$ the \jon ideal of $\m b$ generated by $X$.  The
concept of a \jon ideal was developed in \cite{valeriote-intprop}
for any algebra that generates a congruence distributive variety and
was used in that paper to establish some intersection properties of
subalgebras that are related to relational width.

\begin{df}  A finite algebra $\m b \in \vr v_3$ will be called
\jon trivial if it has no proper non-empty \jon ideals.
\end{df}

Note that $\m b$ is \jon trivial if and only if $J(\{b\}) = B$ for
all $b \in B$.  Also note that if $\m b$ is \jon trivial then every
homomorphic image of it is, as well.

We now define a notion of distance in an algebra that will be
applied to \jon trivial algebras to establish some useful features
of the subalgebras of their cartesian products.

\begin{df}
 Let $\m A$ and $\m B$ be  arbitrary similar algebras and $\m S$
 a subdirect subalgebra  of $\m A\times \m B$.
\begin{enumerate}
\item  Let $S_0 = 0_A$
 and $S_1$  be the relation on $A$ defined by:
\[
\mbox{$(a,c)\in S_1 \iff (a,b), (c,b)\in S$ for some $b\in B$.}
\]

\item  For $k > 0$, let $S_{k+1} = S_k \circ S_1$.

\item For $a$, $b \in A$, we write $d(a,b) = k$ if the pair $(a,b)$
is in $S_k$ and not in $S_{k-1}$ and will say that the distance
between $a$ and $b$ relative to $S$ is $k$. If no such $k$ exists,
$d(a,b)$ is said to be undefined.

\item If $d(a,b)$ is defined for all $a$ and $b \in A$ we say that
$\m a$ is connected with respect to $\m s$.
\end{enumerate}
\end{df}

\begin{prp}
Let $\m a$, $\m b$ and $\m s$ be as in the definition.
\begin{enumerate}
\item
 For each $k \ge 0$, the relation $S_k$ is a reflexive, symmetric
 subuniverse
of $\m a^2$.
 \item  If $\m a$ is an idempotent algebra and $c \in A$ then for
 any $k \ge 0$, the set of all elements $a$ with $d(a,c) \le k$ is a
 subuniverse of $\m a$.
\item If $\m a$ is a simple algebra then either $d(a,b)$ is
undefined for all $a \ne b \in A$ (equivalently $S_1 = 0_A$) or $\m
a$ is connected with respect to $\m s$.
\end{enumerate}
\end{prp}

\begin{proof}
The symmetry of $S_1$ is immediate from its definition and its
reflexivity follows from $S$ being subdirect. To see that it is a
subuniverse of $\m a^2$, let $t(x_1,
\ldots, x_n)$ be a term operation of $\m a$ and $(a_i, b_i) \in S_1$
for $1
\le i
\le n$.  Then for all $i$ there are $c_i \in B$ with $(a_i, c_i)$ and $(b_i,
c_i)
\in S$. Applying $t$ to these pairs shows that $(t(\bar a), t(\bar
c))$ and $(t(\bar b), t(\bar c)) \in S$ and so $(t(\bar a), t(\bar
b)) \in S_1$.  This establishes that $S_1$ is a subuniverse of $\m
a^2$.  Since the relational product operation preserves the
properties of symmetry, reflexivity and being a subuniverse, it
follows that $S_k$ has all three properties, for  $k \ge 0$.

Suppose that $\m a$ is idempotent, $c \in A$, and $k \ge 0$.  If
$t(x_1, \ldots, x_n)$ is a term operation of $\m a$ and $a_i \in A$
with $d(a_i,c) \le k$, for $1 \le i \le n$, then $(a_i, c) \in S_k$
for all $i$.  By the first claim of this proposition, it follows
that $(t(a_1, \ldots, a_n), t(c, \ldots, c)) \in S_k$ since $S_k$ is
a subuniverse of $\m a^2$.  By idempotency we have $t(c, \ldots, c)
= c$ and so $(t(a_1, \ldots, a_n), c) \in S_k$, or $d(t(a_1, \ldots,
a_n),c) \le k$.  This establishes the second claim of the
proposition.

For the last claim, note that since $S_1$ is a symmetric, reflexive
subuniverse of $\m a^2$ then its  transitive closure  is a
congruence on $\m a$ that is equal to the union of the $S_k$, $k
\ge 0$.  Since $\m a$ is assumed to be simple then this congruence
is either $0_A$ or $1_A$.  In the former case we conclude that
$d(a,b)$ is undefined for all $a \ne b \in A$ and in the latter case
that for all $a$, $b \in A$, $(a,b) \in S_k$ for some $k \ge 0$ and
so $d(a,b)$ is defined.
\end{proof}

\begin{lm}\label{distance_inequality}
Let $\m A$ and $\m B$ be  finite algebras in $\vr v_3$ and $\m S$ a
subdirect subalgebra of~$\m A\times\m B$. Suppose that $\m A$ is
connected with respect to $\m s$. Then for every $x,y,z\in A$ we
have
\[
d(x\cdot y, z)\le \max\left(\left[d(x,y)+1\over 2\right],
d(y,z)\right)\,.
\]
\end{lm}

\begin{proof}
Let  $d(y,z) = m$, $d(x,y) = n$ and choose elements $a_i \in A$ for
$0 \le i \le n$ with $x=a_0$, $a_n=y$  and $(a_i, a_{i+1}) \in S_1$
for $0 \le i < n$.  For $k$ the largest integer below $[(n+1)/2]$ we
get that $d(x, a_k)$ and $d(a_k, y)$ are both at most~$k$. Therefore
if $d=\max(k,m)$, then the pairs $(x,a_k), (y,a_k), (y,z)$ are in
$S_d$, and so
\[
(p_2(x,y,y),p_2(a_k,a_k,z))\in S_d\,.
\]
But $p_2(x,y,y)=x\cdot y$ and $p_2(a_k,a_k,z)=z$, proving the lemma.
\end{proof}

\begin{cor}\label{distance_inequality_ideal}
For $\m a$, $\m b$ and $\m s$ as in the previous lemma, suppose that
$d(a,b) \le n$ for all $a$, $b \in A$. Let $m\ge [(n+1)/2]$ be any
integer and $c \in A$.  Then the set of all elements of $A$ whose
distance from $c$ is at most~$m$ is a \jon ideal of $\m a$.
\end{cor}

\begin{proof}
As noted earlier the set $I = \{ a \in A\,:\, d(a,c) \le m\}$ is a
subuniverse of $\m a$ since $\m a$ is idempotent.  We need only show
that $I$ is closed under multiplication on the left.  So, suppose
that $a \in I$ and $u \in A$. Since $d(u,c)\le n$, we have $d(u
\cdot a, c)\le \max(m, d(a,c))\le m$ by the previous lemma.
\end{proof}

\begin{cor}\label{distance_inequality_connected}
Let $\m a$ and $\m b$ be finite members of $\vr v_3$ such that $\m
a$ is  \jon trivial and connected with respect to some subdirect
subalgebra $\m s$ of $\m a \times \m b$. Then $d(a,b) \le 1$ for all
$a$, $b \in A$ (or equivalently, $S_1 = A^2$).
\end{cor}

\begin{proof}
Suppose that the maximum distance~$n$ between the points of~$A$ is
at least~$2$ and that $a$, $b \in A$ with $d(a,b) = n$. Then $m$,
the largest integer below $[(n+1)/2]$ is less than $n$.  From the
previous lemma, the set of all elements $u \in A$ with $d(a,u ) \le
m$ is a proper \jon ideal of $\m a$, contradicting that $\m a$ is
\jon trivial.
\end{proof}

\begin{lm}\label{connected_simple}
Let $\m A$, $\m b $ be finite members of $\vr v_3$ with $\m a$ \jon
trivial and simple and let $\m S$ be a subdirect subalgebra of $\m
A\times \m B$. Then either $S=A\times B$, or $S$ is the graph of an
onto homomorphism from~$\m B$ to~$\m A$.
\end{lm}

\begin{proof}
 As $\m A$ is simple, then either $S_1 = 0_A$ or $\m a$ is
 connected with respect to $\m s$.  In the former case, we conclude
 that $S$ is the graph of an onto homomorphism from $\m b$ to $\m a$
 and in the latter, it follows from the previous corollary that $S_1
 = A^2$.

For $a \in A$, let $B_a = \{b \in B\,:\, (a,b) \in S\}$ and choose
$a$ with $|B_a|$ maximal.   Let $I$ denote the set of those elements
$x$ of $A$ for which $B_x = B_a$.  To complete the proof we will
need to demonstrate that $I = A$ and $B_a = B$.  To show that $I =
A$ it will suffice to prove that it is a \jon ideal of $\m a$.

Indeed, let $u\in A$ and $c\in I$ be arbitrary. Then $(u,c) \in S_1$
 (since $S_1 = A^2$) and therefore
there is a $b\in B$ such that $(u,b)$ and $(c,b)$ are in $S$. Note
that since $c\in I$ then $b\in B_a$. If $d$ is any element of $B_a$
then $c\in I$ implies that $(c,d)\in S$, so we get that
\[
(p_2(u,c,c), p_2(b,b,d)) = (u \cdot c,d)\in S.
\]
Since this holds for every $d\in B_a$, we conclude that $u \cdot c
\in I$.  Finally, since $S$ is subdirect it follows that $B_a=B$.
\end{proof}

We apply this lemma to obtain a simple description of subdirect
products of finite, simple, \jon trivial members of $\vr v_3$ and
then show how to use this description to prove that certain
$k$-minimal instances of the CSP have solutions, when $k \ge 3$.

\begin{lm}\label{full-product}  Let $\m a_i$, for $1 \le i \le n$,
be finite members of $\vr v_3$ with $\m a_1$  \jon trivial.  Let $\m
s$ be a subdirect product of the $\m a_i$'s such that for all $1  <
i \le n$, the projection of $S$ onto coordinates 1 and $i$ is equal
to $A_1 \times \ A_i$. Then $S =A_1 \times D$, where $D = \proj_{\{2
\le i \le n\}}(S)$.
\end{lm}

\begin{proof}
We prove this by induction on $n$.  For $n = 2$, the result follows
by our hypotheses. Consider the case $n = 3$ and let $\m d$ be the
projection of $\m s$ onto $\m a_2 \times \m a_3$. Let $(u,v) \in D$
and let $I_{(u,v)} = \{a \in A_1\,:\, (a,u,v) \in S\}$. Our goal is
to show that $I_{(u,v)} = A_1$ and we can accomplish this by showing
that it is a non-empty \jon ideal. Clearly $I_{(u,v)}$ is a
non-empty subuniverse of $\m a_1$ since all algebras involved are
idempotent.

Let $a \in I_{(u,v)}$, $b \in A_1$ and choose elements
 $y \in A_3$ and $x \in A_2$ with $(b,u,y)$ and $(a,x,y)
\in S$.  By our hypotheses, these elements exist. Applying $p_2$
 to these elements, along with $(a,u,v)$, we get  the
element $(b \cdot a, u, v)$, showing that $b \cdot a \in I_{(u,v)}$.
Thus $I_{(u,v)}$ is a \jon ideal.

Now, consider the general case and suppose that the result holds for
products of fewer than $n$ factors.  Let $\m s_1 = \proj_{\{1 \le i
< n\}}(S)$ and $\m s_2 = \proj_{\{2 \le i <n\}}(S)$.  Then $\m s$ is
isomorphic to a subdirect product of $\m a_1$, $\m s_2$ and $\m a_n$
and, by induction, $S_1 = A_1 \times S_2$.  Then, applying the
result with $n = 3$ to this situation, we conclude that $S = A_1
\times D$, as required.
\end{proof}

\begin{cor}\label{trivial-product}
Let $\m a_i$ be finite, simple, \jon trivial members of $\vr v_3$,
for $1 \le i \le n$, and let $\m s$ be a subdirect product of the
$\m a_i$'s. If, for all $1 \le i < j \le n$, the projection of $S$
onto $A_i \times A_j$ is not the graph of a bijection then $S =
\prod_{1 \le i \le n} A_i$.
\end{cor}

\begin{proof}
For $1 \le i < j \le n$, we have, by Lemma \ref{connected_simple}
that either the projection of $S$ onto $A_i \times A_j$ is the graph
of a bijection between the two factors (since they are both simple)
or is the full product.  The former case is ruled out by assumption
and so we are in a position to apply the previous lemma inductively
to reach the desired conclusion.
\end{proof}

\begin{df}\label{almost-triv}
A subdirect product $\m s$ of the algebras $\m a_i$, $1 \le i \le
n$, is said to be almost trivial if, after suitably rearranging the
coordinates, there is a partition of $\{1, 2, \ldots, n\}$ into
intervals $I_j$, $1 \le j \le p$, such that $S = \proj_{I_1}(S)
\times \cdots \times \proj_{I_p}(S)$ and, for each $j$, if $I_j =
\{i\,:\, u \le i \le v\}$ then there are bijections $\pi_i:A_u
\rightarrow A_i$, for $i \in I_j$ such that $\proj_{I_j}(S) =\{(a,
\pi_{u+1}(a), \ldots, \pi_v(a))\,:\, a \in A_u\}$.
\end{df}

\begin{cor}\label{almost-trivial}
 Let $\m a_i$ be finite, simple, \jon trivial members of $\vr
v_3$, for $1 \le i \le n$, and let $\m s$ be a subdirect product of
the $\m a_i$'s.  Then $\m s$ is almost trivial.
\end{cor}

\begin{proof}
For $1 \le i,j \le n$, set $i\sim j$ if $i = j$ or the projection of
$S$ onto $A_i$ and $A_j$ is equal to the graph of a bijection
between these two factors.  In this case, let $\pi_{i,j}$ denote
this bijection.

It is not hard to see that  $\sim $ is an equivalence relation on
the set $\{1, 2, \ldots, n\}$ and, by applying Lemma
\ref{connected_simple}, if $i \not\sim j$ then the projection of $S$
onto $A_i$ and $A_j$ is equal to $A_i \times A_j$.  By using the
bijections $\pi_{i,j}$ and Corollary \ref{trivial-product} it is
elementary to show that $\m s$ is indeed almost trivial.
\end{proof}

For $\vr a$ a finite sequence of finite algebras,  $P = (\vr a, \vr
c)$ denotes a multi-sorted instance of the CSP whose domains are the
universes of the algebras in $\vr a$ and whose constraint relations
are subuniverses of cartesian products of members from $\vr a$.

\begin{thm}\label{induction-base}
Let $\vr a$ be a finite sequence of finite, simple, \jon trivial
members of $\vr v_3$ and let $P = ( \vr a, \vr c)$ be a subdirect,
$k$-minimal instance of the CSP for some $k\ge 3$. If the constraint
relations of $P$ are all non-empty then $P$ has a solution.
\end{thm}

Definition \ref{almost-triv} and analogs of  Corollary
\ref{almost-trivial} and Theorem \ref{induction-base} can be found
at the end of Section 3.3 in \cite{bulatov-2semi}. The proof of
Corollary 3.4 given in that paper can be used to prove our Theorem
\ref{induction-base}.  As we shall see, this theorem will form the
base of the inductive proof of our main result.

\subsection{The reduction to \jon trivial algebras}

The goal of this subsection is to show how to reduce a $k$-minimal
instance $P$ of the CSP whose domains all lie in $\vr v_3$ and whose
constraint relations are all non-empty to another $k$-minimal,
subdirect instance $P'$ whose domains are all \jon trivial and whose
constraint relations are non-empty.  In order to accomplish this, we
will need to work with a suitably large $k\ge 3$.

To start, let $\vr a = (\m a_1, \ldots, \m a_n)$ be a sequence of
finite algebras from $\vr v_3$  and let $M= \max\{|A_i|\,:\, 1 \le i
\le n\}$. Let $k
> 0$  and $P = (\vr a, \vr c)$ be a $k$-minimal instance of the CSP
 with $\vr c$ consisting of
the constraints $C_i = (S_i, R_i)$, $1 \le i \le m$.  By taking
suitable subalgebras of the $\m a_i$ we may assume that $P$ is
subdirect and, of course,  we also assume that the $R_i$ are all
non-empty.  In addition, $k$-minimality assures that we may assume
that the scope of each constraint of $P$ consists of at least $k$
variables and that no two constraints have the same $k$-element set
as their scopes.

Since $P$ is $k$-minimal then its system of partial solutions over
$k$-element sets satisfies an important compatibility property.
Namely, if $I$ and $K$ are $k$-element sets of coordinates then
$\proj_{(I\cap K)}(P_I) = \proj_{(I\cap K)}(P_K)$.  In this section
we will denote $P_I$ by $\Lambda(I)$ and call this function the
$k$-system (of partial solutions) determined by $P$.  Since $P$ is
subdirect then for all $I$, $\Lambda(I)$ will be a subdirect product
of the algebras $\m a_i$, for $i \in I$.

We wish to consider the situation in which some $\m a_i$, say $\m
a_1$, has a proper \jon ideal $J$.  The  main result of this
subsection is that if the scopes of the constraints of $P$ all have
size at most $k$ (and hence exactly $k$), or if $k \ge M^2$ then we
can reduce the question of the solvability of $P$ to the solvability
of a $k$-minimal instance with $\m a_1$ replaced by $J$. Doing so
will allow us to proceed by induction to reduce our original
instance down to one whose domains are all \jon trivial.

So, let $J$ be a proper non-empty \jon ideal of $\m a_1$ and define
$\Lambda_J$ to be the following function on the set of $k$-element
subsets of $\{1, 2, \ldots, n\}$:
\begin{itemize}
\item If $I$ is a $k$-element set that includes 1 then define
$\Lambda_J(I)$ to be $\{ \vec{a}\in \Lambda(I)\,:\, \vec{a}(1) \in
J\}$.

\item If $1 \notin I$, define $\Lambda_J(I)$ to be the set of all $\vec{a} \in
\Lambda(I)$ such that for all $i \in I$ the restriction of $\vec{a}$
to $I\setminus \{i\}$ can be extended to an element of
$\Lambda_J(\{1\} \cup (I\setminus\{i\}))$.
\end{itemize}

\begin{lm}\label{Gamma_J}
 If $k \ge 3$ then
\begin{enumerate}
\item $\Lambda_J(I)$ is non-empty for all $I$ and
 if $1 \in I$ then the projection of $\Lambda_J(I)$ onto the
first coordinate is equal to $J$.
\item  For $I$, $K$, $k$-element subsets of $\{1, 2, \ldots, n\}$,
$\proj_{(I\cap K)}(\Lambda_J(I)) = \proj_{(I\cap K)}(\Lambda_J(K))$.
\end{enumerate}
\end{lm}

\begin{proof}
Since $P$ is subdirect then for any $k$-element set $I$ with $1 \in
I$ we have that $\Lambda_J(I)$ is non-empty and projects onto $J$ in
the first coordinate.

Let $I$ be some $k$-element set of coordinates with $1 \notin I$.
For ease of notation, we may assume that $I = \{2, 3, \ldots, k,
k+1\}$.  Let $\vec{a} = (a_1, a_2, a_3, \ldots, a_k)$ be any member
of $\Lambda_J(\{1,2, \ldots, k\})$.  We will show that there is some
$a_{k+1} \in A_{k+1}$ such that $(a_2, a_3, \ldots, a_{k+1}) \in
\Lambda_J(I)$.  This will not only show that $\Lambda_J(I)$ is
non-empty, but will also allow us to easily establish  condition (2)
of the lemma.

We construct the element $a_{k+1}$ as follows.  Since $\Lambda$ is
the $k$-system  for $P$ then there is some element $u \in A_{k+1}$
such that
 $( a_2, \ldots, a_k,u) \in \Lambda(I)$. Furthermore,
there is some $v\in A_{k+1}$ such that $(a_1, a_3, \ldots, a_k, v)
\in \Lambda(\{1,3,\ldots, k+1\})$ and then some $v' \in A_2$ with
$(v', a_3, \ldots, a_k, v) \in \Lambda(I)$.  Similarly, there are
$w$ and $w'$ with $(a_1, a_2, a_4, \ldots, a_k, w) \in
\Lambda\{1,2,4, \ldots, k+1\})$ and $(a_2, w', a_4, \ldots, a_k, w)
\in \Lambda(I)$. Let $a_{k+1} = p_1(u,v,w)\in A_{k+1}$.  By applying
  $p_1$ to the tuples $( a_2, \ldots, a_k,u)$, $(v', a_3, \ldots,
a_k,v)$ and $(a_2, w', a_4, \ldots, a_k, w)$ we see that the tuple
$(a_2, a_3, \ldots, a_{k+1}) \in \Lambda(I)$.

We now need  to show that for all $2 \le i \le k+1$ there is some $b
\in J$ with
\[(b, a_2, \ldots, a_{i-1}, a_{i+1}, \ldots, a_{k+1}) \in
\Lambda_J(\{1,2, \ldots, i-1, i+1, \ldots, k+1\}).\]
There are a
number of cases to consider.
\begin{itemize}
\item If $i = k+1$ then the tuple $(a_2, \ldots, a_k)$ extends to $(a_1, a_2,
\ldots, a_k)$, a member of $\Lambda_J(\{1, 2, \ldots, k\})$, as
required.

\item If $i = 2$: There are  $x \in A_1$ and $y \in A_3$ with
$(x, a_3, \ldots, a_k, u)$ and $(a_1, y, a_4, \ldots, a_k, w)$ in
$\Lambda(\{1,3,
\ldots, k+1\})$.  Applying $p_1$ to these tuples and the
tuple $(a_1, a_3, a_4, \ldots, a_k, v)$ (in the second variable)
produces the tuple $(x\cdot a_1, a_3, \ldots, a_k, a_{k+1}) \in
\Lambda(\{1,3, \ldots, k+1\})$.  Since $a_1 \in J$ and $J$ is a \jon
ideal, then $x\cdot a_1 \in J$ and so this tuple belongs to
$\Lambda_J(\{1,3, \ldots, k+1\})$, as required.

\item If $i = 3$ or $3 < i < k+1$ then  small variations of the previous argument will
work.

\end{itemize}

To complete the proof of this lemma we need to establish the
compatibility of $\Lambda_J$ on overlapping elements of its domain.
Let $I$ and $L$ be distinct members of the domain of $\Lambda_J$
with non-empty intersection $N$ and let $i \in I\setminus L$ and $l
\in L \setminus I$.

Let $\vec{a} \in \Lambda_J(I)$ and let $\vec{c}$ be the projection
of $\vec{a}$ onto the coordinates in $N$.  The restriction of
$\vec{a}$ to $I \setminus \{i\}$ extends to an element $\vec{a}' \in
\Lambda_J(\{1\} \cup (I \setminus\{i\}))$.  Since $\Lambda$ is the
$k$-system for $P$, the restriction of $\vec{a}'$ to $\{1\}\cup N$
extends to an element $\vec{b}' $ of $\Lambda(\{1\}\cup (L
\setminus\{l\}))$. Note that $\vec{b}'(1) \in J$ and the restriction
of $\vec{b}'$ to $N$ is $\vec{c}$. By the first part of this proof,
it follows that the restriction of $\vec{b}'$ to $L \setminus \{l\}$
extends to an element $\vec{b}$ of $\Lambda_J(L)$ as required.
\end{proof}

\begin{cor}\label{local} If all of the constraints of $P$ have scopes of size
$k$ then there is a $k$-minimal instance $P_J$ of the constraint
satisfaction problem over the domains $\m J$ and the $\m a_i$, for
$2 \le i \le n$,  whose constraint relations are all non-empty and
whose solution set is contained in the solution set of $P$.
\end{cor}

\begin{proof} It follows from our assumptions on
the sizes of the scopes of the constraints of $P$ that the
constraints can be indexed by the $k$-element subsets of $\{1, 2,
\ldots, n\}$ and that for such a subset $I$, the constraint $C_I$ is
of the form $(I, R_I)$ where $R_I$ is a subdirect product of the
algebras $\m a_i$, for $i \in I$.

We set $P_J$ to be the instance of the CSP over the domains $\m j$
and the $\m a_i$, for $2 \le i \le n$, that has, for each
$k$-element subset $I$ of $\{1, 2, \ldots, n\}$, the constraint
$C'_I = (I, R'_I)$, where $R'_I = \Lambda_J(I)$. It follows by
construction and from the previous lemma that $P_J$ is a $k$-minimal
instance of the CSP whose constraint relations are all non-empty and
whose solutions are also solutions of $P$.
\end{proof}

The previous corollary can be used to establish the tractability of
the constraint languages arising from finite members of $\vr v_3$,
while the following lemma will be used to prove that these languages
are in fact globally tractable.

\begin{lm}\label{jon-reduction}
Assume that $k \ge M^2$ and let $C = (S, R)$ be a constraint of $P$.
Then there is a subuniverse $R_J$ of $R$ such that for all
$k$-element subsets $I$ of $S$, the projection of $R$ onto $I$ is
equal to $\Lambda_J(I)$.
\end{lm}

\begin{proof}
For $K$ a subset of $S$ and $\vec{a} \in R$, we will say that
$\vec{a}$ is reduced over $K$ if for all $(k-1)$-element subsets $I$
of $K$, the restriction of $\vec{a}$ to $I$ can be extended to an
element of $\Lambda_J(\{1\}\cup I)$. We define $R_J$ to be the set
of all tuples $\vec{a} \in R$ that are reduced over $S$. $R_J$ is
also equal to all elements $\vec{a}$ of $R$ such that for all
$k$-element subsets $I$ of $S$, the restriction of $\vec{a}$ to $I$
is in $\Lambda_J(I)$. $R_J$ is naturally a subuniverse of $R$ and so
the challenge is to show that it satisfies the conditions of the
lemma. Our proof breaks into two cases, depending on whether or not
the coordinate 1 is in $S$.

Suppose that $1 \in S$.  We may assume that $S= \{ 1, 2, \ldots,
m\}$ for some $m\le n$.  We need to show that if $I$ is a $k$
element subset of $S$ and $\vec{a} \in \Lambda_J(I)$ then there is
some $\vec{b} \in R_J$ whose restriction to $I$ is $\vec{a}$.

First consider the sub-case where $1 \in I$.  If $\vec{a} \in
\Lambda_J(I)$ then by the $k$-minimality of $P$ there is some
$\vec{b} \in R$ whose restriction to $I$ is $\vec{a}$. Since
$\vec{b}(1) = \vec{a}(1 )\in J$  it follows that $\vec{b}$ is in
$R_J$, as required.

Now, suppose that $1 \notin I$ and assume that $I = \{2, 3,\ldots,
k+1\}$.  By the $k$-minimality of $P$ there is some $\vec{c} \in R$
whose restriction to $I$ is $ \vec{a}$.  For each $2 \le i \le k+1$
there is some $j_i \in J$  and some $\vec{c}_i \in R$ such that
$\vec{c}_i(1) = j_i$ and such that the restrictions of $\vec{c}_i$
and $\vec{a}$ to $I \setminus \{i\}$ are the same.

Since $k > |J|$ it follows from the Pigeonhole principle that there
are $i \ne l$ with $j_i = j_l$. We may assume that $i = 2$ and $l =
3$ and set $j = j_i$.  Define $\vec{b}$ to be $p_1(\vec{c},
\vec{c}_2, \vec{c}_3)$.  This element belongs to $R$ and satisfies:
$\vec{b}(1) = \vec{c}(1 )\cdot j \in J$ and the restriction of
$\vec{b}$ to $I$ is $\vec{a}$. To establish this equality over
coordinate 2 we make use of the identity $p_1(x,y,x) = x$ and over
coordinate 3 $p_1(x,x,y) = x$. Finally, $\vec{b}$ is in $R_J$ since
$\vec{b}(1 )\in J$.

For the remaining case, assume that $1 \notin S$, say $S = \{ 2,3,
\ldots, m+1\}$. We will show by induction on $s$ that if $k-1 \le s
\le m-1$, $K$ is a subset of $\{2, 3, \ldots, m+1\}$ of size $s$ and
$\vec{a} \in R$ is reduced over $K$ then if $i \in S\setminus K$
there is some $\vec{b} \in R$ that is reduced over $K \cup \{i\}$
and such that $\proj_K(\vec{a}) = \proj_K(\vec{b})$.
 A consequence of this claim is that for any
$k$-element subset $I$ of $S$, any element of $\Lambda_J(I)$ can be
extended to a member of $R_J$. From this, the lemma follows.

Lemma \ref{Gamma_J} establishes the base of this induction. Assume
the induction hypothesis holds for $k -1 \le s < m -1$ and let $K$
be a subset of $\{2, 3, \ldots, m+1\}$ of size $s+1$. By symmetry,
we may assume that $K = \{2, 3, \ldots, s+2\}$.  Let $\vec{a} \in R$
be reduced over $K$. We will show that there is some $\vec{a}' \in
R$ which equals $\vec{a}$ over $K$ and is reduced over $K \cup
\{s+3\}$.

By the induction hypothesis, for each $2 \le i \le s+2$  there is
some $\vec{a}_i \in R$ such that the projections of $\vec{a}$ and
$\vec{a}_i$ onto $K \setminus \{i\}$ are the same and $\vec{a}_i$ is
reduced over $(K \cup\{s+3\}) \setminus \{i\}$.   By the Pigeonhole
principle it follows that there is some $a \in A_{s+3}$ and  a set
$Q$ contained in $K$ of size at least $M$ such that for  $i \in Q$,
$\vec{a}_i(s+3) = a$.

Let $i$ and $l$ be distinct members of $Q$ and let $\vec{a}'$ be the
element $p_1(\vec{a}, \vec{a}_{i}, \vec{a}_{l})$ of $R$. Note that
over the coordinates in $K$, $\vec{a}'$ and $\vec{a}$ are equal and
that at $s+3$, $\vec{a}'$ equals $b\cdot a$, where $ b =
\vec{a}(s+3)$.

We claim that $\vec{a}'$ is reduced over $K \cup \{ s+3\}$. To
establish this we need to show that over any subset $U$ of $K \cup
\{ s+3\}$  of size $k-1$, the restriction to $U$ of $\vec{a}'$ can
be extended to a member of $\Lambda_J( \{1\}\cup U)$. When $U$
avoids the coordinate $s+3$ there is nothing to do, since $\vec{a}$
is reduced over $K$.

So, assume that $U$ contains $s+3$ and let $\vec{d}$ be an extension
to some element in $\Lambda(\{1\}\cup U)$  of the restriction of
$\vec{a}$ to $U$.  Since for each $v \in Q$ the element $\vec{a}_v$
is reduced over $(K\cup\{s+3\})\setminus \{v\}$ then there is a
member $\vec{c}_v$ of $\Lambda_J(\{1\}\cup U)$ whose restriction to
$U \setminus \{v\}$ is equal to the restriction of $\vec{a}_v$ over
this set.  If there is some $v \in Q \setminus U$ then the element
$p_1(\vec{d}, \vec{c}_v, \vec{c}_v) \in \Lambda_J(\{1\}\cup U)$
witnesses that the restriction of $\vec{a}'$ to $U$ can be extended
as desired.

If, on the other hand, $Q \subseteq U$ then choose two elements $u$
and $v$ of $Q$ such that $\vec{c}_u(1) = \vec{c}_v(1) \in J$.  An
application of the Pigeonhole principle ensures the existence of
these elements since $|Q| > |J|$.  Then, the element $p_1(\vec{d},
\vec{c}_u, \vec{c}_v) \in \Lambda_J(\{1\}\cup U)$ and its
restriction to $U$ is equal the restriction of $\vec{a}'$ on $U$.
\end{proof}

\begin{cor}\label{global}
If $k \ge M^2$ then there is a $k$-minimal instance $P_J$ of the
constraint satisfaction problem over $\m J$ and the $\m a_i$, for $2
\le i \le n$, whose constraint relations are all non-empty and whose
solution set is contained in the solution set of $P$.
\end{cor}

\begin{proof} From the preceding lemma it follows that the instance
$P_J$  over the domains $\m j$ and the $\m a_i$, for $2 \le i \le
n$, with constraints $C' = (S, R_J)$, for each constraint $C = (S,
R)$ of $P$, is $k$-minimal and has all of its constraint relations
non-empty.  Since the constraint relations of $P_J$ are subsets of
the corresponding constraint relations of $P$ then the result
follows.
\end{proof}

\begin{thm} \label{jon-trivial}
Let $\vr a  = (\m a_1, \ldots, \m a_n)$ be a  sequence of finite
algebras  from $\vr v_3$ and let $P = (\vr a, \vr c)$ be a
$k$-minimal instance of the CSP whose constraint relations are
non-empty. If $k \ge 3$ and the sizes of the scopes of the
constraints of $P$ are bounded by $k$ or if $k \ge M^2$, where $M=
\max\{|A_i|\,:\, 1 \le i \le n\}$, then there is a subdirect
$k$-minimal instance $P'$ of the CSP over  \jon trivial subalgebras
of the $\m a_i$ such that the constraint relations of $P'$ are
non-empty and the solution set of $P'$ is contained in the solution
set of $P$.
\end{thm}

\begin{proof}  This theorem is proved by repeated application of
Corollaries \ref{local} and \ref{global}.
\end{proof}

\subsection{The reduction to simple algebras}

In this subsection we show, for $k \ge 3$, how to reduce a
$k$-minimal instance of the CSP whose domains are \jon trivial
members of $\vr v_3$  and whose constraint relations are all
non-empty to one which has in addition, domains that are simple
algebras.  Our development closely follows parts of the proof of
Theorem 3.1 in \cite{bulatov-2semi}.

\begin{df}  Let $\m a_i$, $1 \le i \le m$, be similar algebras and
let $\Theta = (\theta_1, \ldots, \theta_m)$ be a sequence of
congruences $\theta_i \in \cn(\m a_i)$.
\begin{enumerate}
\item $\prod_{i = 1}^{m} \theta_i$ denotes the congruence on $\prod_{i =
1}^{m} \m a_i$ that identifies two $m$-tuples $\vec{a}$ and
$\vec{b}$ if and only if $(a_i, b_i) \in \theta_i$ for all $i$.

\item
If $I$ is a subset of $\{1, 2, \ldots, m\}$ and $\m r$ is a
subalgebra of $\prod_{i \in I}\m a_i$ then $\m r/\Theta$ denotes the
quotient of $\m r$ by the restriction of the congruence $\prod_{i
\in I}\theta_i$ to $R$.
\end{enumerate}
\end{df}

Let $\vr a = (\m a_1, \ldots, \m a_n)$ be a sequence of finite, \jon
trivial members of $\vr v_3$ and let $P = (\vr a, \vr c)$ be a
subdirect, $k$-minimal instance of the CSP whose constraint
relations are all non-empty. Let $\vr c = \{C_1, C_2, \ldots, C_m\}$
where, for $1 \le i \le m$, $C_i = (S_i, R_i)$ for some subset $S_i$
of $\{1,2, \ldots, n\}$ and some subuniverse $R_i$ of $\prod_{i \in
S_i}\m a_i$.  Suppose that one of the $\m a_i$ is not simple, say
for $i = 1$, and let $\theta_1 $ be a maximal proper congruence of
$\m a_1$.

Recall that for $I \subseteq \{1, 2, \ldots, n\}$, $P_I$ denotes the
set of partial solutions of $P$ over the variables $I$.  If $|I| \le
k$ then since $P$ is $k$-minimal, $P_I$ is non-empty and is a
subdirect subuniverse of $\prod_{i \in I}\m a_i$.

 Since the algebra
$\m a_1/\theta_1$ is a simple, \jon trivial algebra then it follows
by Lemma \ref{connected_simple} that for $2 \le i \le n$,
$P_{\{1,i\}}/(\theta_1 \times 0_{A_i})$ is either the graph of a
homomorphism $\pi_i$ from $\m a_i$ onto $\m a_1/\theta_1$ or is
equal to $ A_1/\theta_1 \times A_i$.  Let $W$ consist of $1$ along
with the set of all $i$ for which the former holds. For $2 \le i \le
n$, let $\theta_i$ be the kernel of the map $\pi_i$ if $i
 \in W$, and $0_{A_i}$ otherwise.

Let $\Theta = (\theta_1, \ldots, \theta_n)$ and set $P/\Theta = (\vr
a/\Theta, \vr c/\Theta)$ where $\vr a/\Theta = (\m a_1/\theta_1,
\ldots, \m a_n/\theta_n)$ and $\vr c/\Theta$ consists of the
constraints $C_i/\Theta = (S_i, R_i/\Theta)$, for $ 1\le i \le m$.

Note that since $P$ is subdirect and $k$-minimal then so is
$P/\Theta$ and that each $\m a_i /\theta_i$ is \jon trivial, since
this property is preserved by taking quotients.

\begin{lm}\label{pullback}  If the instance $P/\Theta$ has a solution, then there is
some $k$-minimal instance $P' = (\vr a', \vr c')$ such that
\begin{itemize}
\item $\vr a' = (\m a'_1, \ldots, \m a'_n)$, where for each $1 \le i
\le n$,  $\m a'_i$ a subalgebra of $\m a_i$.
\item $A'_1 $ is a proper subset of $A_1$,
\item $\vr c' = \{C'_1, \ldots, C'_m\}$ where, for each $1 \le i \le m$,
 $C'_i = (S_i, R'_i)$ for
some non-empty subuniverse $R'_i$ of $R_i$.
\end{itemize}
Hence, any solution of $P'$ is a solution of $P$.
\end{lm}

\begin{proof}
Let $(s_1, \ldots, s_n)$ be a solution of $P/\Theta$.  We can regard
each  $s_i$  as a congruence block of $\theta_i$  and hence as a
subuniverse  of $\m a_i$. For $i \in W$, define $\m a'_i$ to be the
subalgebra of $\m a_i$ with universe $s_i$ and for $i \notin W$, set
$\m a'_i = \m a_i$.  For $1 \le j \le m$, let
\[
R'_j = R_j \cap  \prod_{i \in S_j} A'_i.
\]

We now set out to prove that the instance $P' = (\vr a', \vr c')$
has the desired properties.  Since $\theta_1$ is a proper congruence
of $\m a_1$ then $s_1$ is a proper subset of $A_1$ and so $A'_1$ is
properly contained in $A_1$.  Since $(s_1, \ldots, s_n)$ is a
solution to $P/\Theta$ it follows that for $1 \le j \le m$,  $R'_j$
is a non-empty subuniverse of $R_j$.

We need only  verify that $P'$ is $k$-minimal, so  let $1 \le a < b
\le m$ and $I$ be some subset of $S_a \cap S_b$ of size at most $k$.
To establish that $\proj_I(R'_a) = \proj_I(R'_b)$ it will suffice to
show that
\[
\proj_{I } (R'_i) = \proj_{I}(R_i) \cap \prod_{l\in I}A'_l.
\]
for all $i$, since $P$ is $k$-minimal.

By the definition of $R'_i$ it is immediate that the relation on the
left of the equality sign is contained in that on the right.  In the
case that $W \cap S_i= \emptyset$  the other inclusion is also
clear.

If $W\cap S_i \ne\emptyset$ we have that $\proj_{W\cap
S_i}(R_i/\Theta)$ is a subdirect product of simple, \jon trivial
algebras that are all isomorphic to $\m a_1/\theta_1$.  Since the
projection of this subdirect product onto any two coordinates in
$W\cap S_i$ is equal to the graph of a bijection then in fact, the
entire subdirect product is isomorphic to $\m a_1/\theta_1$ in a
natural way (using the bijections $\pi_i$ from the definition of
$W$). Then, using Lemma \ref{full-product} and the definition of $W$
(or more precisely, the complement of $W$), we conclude that
$R_i/\Theta$ is isomorphic to $\m a_1/\theta_1 \times D$, where $D =
\proj_{(S_i\setminus W)}(R_i)$.

Now, suppose that $\vec{a} \in \proj_{I}(R_i) \cap \prod_{l\in
I}A'_l$.  Then there is some $\vec{b} \in R_i$ with
$\proj_{I}(\vec{b}) = \vec{a}$.  If $W\cap I = \emptyset$ then, by
the concluding remark of the previous paragraph,
$\proj_{(S_i\setminus W)}(\vec{b})$ and hence $\proj_{I}(\vec{b})$
can be extended to an element of $R_i$ that lies in $\prod_{l \in
S_i}A'_l$ (here we use the fact that we have a solution of
$P/\Theta$ to work with). This establishes that, in this case,
$\vec{a} \in \proj_I(R'_i)$.

Finally, suppose that for some $w$ we have  $w \in W\cap I$.  The
vector $\vec{b}$ from $R_i$ that projects onto $\vec{a}$ over $I$
has the property that  $\vec{b}(w) \in s_w$ (since $\vec{a}$ does).
The structure of $R_i/\Theta $ worked out earlier implies that
$\vec{b}(l) \in s_l$ for all $l \in W\cap S_i$ since $(s_1, \ldots,
s_n)$ is a solution to $P/\Theta$. From this we conclude that
$\vec{b} \in R'_i$, as required.
\end{proof}

\section{Proof of the main result}

In the preceding section we established techniques for reducing
$k$-minimal instances of the CSP over domains from $\vr v_3$ to more
manageable instances.  The following theorem employs these
techniques to establish the finite relational width of constraint
languages arising from finite algebras in $CD(3)$.

Let $\m a$ be a finite algebra in $CD(3)$.  Then $\m a$ has term
operations $p_1(x,y,z)$ and $p_2(x,y,z) $ that satisfy the
equations:
\begin{eqnarray*}
p_i(x,y,x) &= &x \ \ \mbox{, $i = 1, 2$}\\
p_1(x,x,y) &=& x\\
p_1(x,y,y) &=& p_2(x,y,y)\\
p_2(x,x,y) &= &y
\end{eqnarray*}

Recall that associated with $\m a$ is the constraint language
$\Gamma_{\sm a} =
\inv(\m a)$, consisting of all relations invariant under the basic
operations of $\m a$.

\begin{thm}\label{main-result}
If $\Gamma$ is a subset of $\Gamma_{\sm a}$ whose relations all have
arity $k$ or less, for some $k \ge 3$, then $\Gamma$ has relational
width $k$.  In any case, if $M = |A|^2$ then $\Gamma_{\sm a}$ has
relational width $M$.
\end{thm}

\begin{cor}
If $\Gamma$ is a finite subset of $\Gamma_{\sm a}$ then $\Gamma$ is
tractable and is of bounded width in the sense of Feder-Vardi.
Furthermore, $\Gamma_{\sm a}$ is globally tractable.
\end{cor}

\begin{proof} (of the Theorem)
We may assume  that $\m a = (A, p_0, p_1, p_2, p_3)$, where
$p_0(x,y,z) = x$ and $p_3(x,y,z)  = z$ for all $x$, $y$, $z \in A$
since if we can establish the theorem for this sort of algebra, it
will then apply to all algebras with universe $A$ that have the
$p_i$ as term operations.

Our assumption on $\m a$ places it in the variety $\vr v_3$ and so
the results from the previous section apply.  Let $\Gamma$ be a
subset of $\Gamma_{\sm a}$.  If $\Gamma$ is finite, let $k$ be the
maximum of 3 and the arities of the relations in $\Gamma$ and
replace $\Gamma$ by $\Gamma_k$, the set of all relations in
$\Gamma_{\sm a}$ of arity $k$ or less.  Establishing relational
width $k$ for this enlarged $\Gamma$ will, of course, be a stronger
result. If $\Gamma$ is not finite, replace it by $\Gamma_{\sm a}$
and set $k = |A|^2$.  We will show that in either case, $\Gamma$ has
relational width $k$.

From statements (5) or (6) of Proposition \ref{width-prop} it will
suffice to show that if $P$ is a $k$-minimal instance of
 $CSP(\Gamma)$ whose constraint relations are all non-empty then $P$
 has a solution.
We may  express $P$ in  the form $(\vr a, \vr c)$ where $\vr a = (\m
a, \m a, \ldots, \m a)$ is a sequence of length $n$, for some $n >
0$, and where $\vr c $ is a set of  constraints  of the form $C =
(S, R)$, for some non-empty subset $S$  of $\{1,2, \ldots, n\}$ and
some non-empty subuniverse $R$ of $\m a^{|S|}$.

In order to apply the results from the previous section as
seamlessly as possible, we enlarge our language $\Gamma$ to a
closely related, but larger, multi-sorted language.  Let $\vr h$ be
the set of all quotients of subalgebras of $\m a$.  Note that $\vr
h$ is finite and all algebras in it have size at most $|A|$. If
$\Gamma = \Gamma_k$, replace it with the set of all subuniverses of
$l$-fold products of algebras from $\vr h$, for all $1 \le l \le k$,
and otherwise, replace it by the set of all subuniverses of finite
products of algebras from $\vr h$.  In both cases, we have extended
our original constraint language.  $P$ can now be viewed as a
$k$-minimal instance of $CSP(\Gamma)$, the class of multi-sorted
CSPs whose instances have domains  from $\vr h$ and whose constraint
relations are from $\Gamma$.

We now prove that every $k$-minimal instance of $CSP(\Gamma)$ whose
constraint relations are non-empty has a solution.  If this is not
so, let $Q$ be a counter-example such that the sum of the sizes of
the domains of $Q$ is as small as possible.  Note that independent
of this size, no domain of $Q$ is bigger than $|A|$ since they all
come from $\vr h$.  Also note that $Q$ must be subdirect.

From Theorem \ref{jon-trivial} it follows that all of the domains of
$Q$ are \jon trivial.  Then, from Lemma \ref{pullback} we can deduce
that all of the domains of $Q$ are simple.  If not, then either
there is a proper quotient of $Q$ that is $k$-minimal and that does
not have a solution, or the $k$-minimal instance produced by the
lemma cannot have a solution. In either case, we contradict the
minimality of $Q$.  Thus $Q$ is a subdirect, $k$-minimal instance of
$CSP(\Gamma)$ whose domains are all simple and \jon trivial and
whose constraint relations are all non-empty.  From Theorem
\ref{induction-base} we conclude that in fact $Q$ has a solution.
This contradiction completes the proof of the theorem.
\end{proof}

\section{Conclusion}

The main result of this paper establishes that for certain
 constraint languages $\Gamma$ that arise from
finite algebras that generate congruence distributive varieties, the
problem class $CSP(\Gamma)$ is tractable. This class of constraint
languages includes those that are compatible with a majority
operation but also includes some languages that were not previously
known to be tractable.

We feel that the proof techniques employed in this paper may be
useful in extending our results to include all constraint languages
that arise from finite algebras that generate congruence
distributive varieties and perhaps beyond.

\medskip
{\noindent \bf Problem 1:}  Extend the algebraic tools developed to
handle algebras in $CD(3)$ to algebras in $CD(n)$ for any $n > 3$.
In particular, generalize the notion of a \jon ideal to this wider
setting.

\medskip

We note that in \cite{valeriote-intprop} some initial success at
extending the notion of a \jon ideal has been obtained.

The bound on relational width established for the languages
addressed in this paper seems to depend on the size of the
underlying domain of the language. Nevertheless, we are not aware of
any constraint language that has finite relational width that is not
of relational width 3.

\medskip
{\noindent \bf Problem 2:}  For each $n > 3$, produce a constraint
language $\Gamma_n$ that has relational width $n$ and not $n-1$. As
a strengthening of this problem, find  $\Gamma_n$ that in addition
have compatible near unanimity operations.

\section{Acknowledgments}

The first author acknowledges the support of the Hungarian National
Foundation for Scientific Research (OTKA), grants no. T043671 and
T043034, while the second, the support of the Natural Sciences and
Engineering Research Council of Canada.  Support of the Isaac Newton
Institute for  Mathematical Sciences and the organizers of the Logic
and Algorithms programme is also gratefully acknowledged.

\bibliographystyle{plain}

%\bibliography{ua}

\end{document}